\documentclass{article}
\usepackage{amsmath,amssymb}
\usepackage{algorithm,algorithmic}
\newtheorem{thm}{Theorem}
\newtheorem{lem}[thm]{Lemma}
\newtheorem{cor}[thm]{Corollary}
\newtheorem{definition}[thm]{Definition}

\newenvironment{proof}[1][Proof]{\begin{trivlist}
\item[\hskip \labelsep {\bfseries #1}]}{\end{trivlist}}

\newcommand{\qed}{\nobreak \ifvmode \relax \else
      \ifdim\lastskip<1.5em \hskip-\lastskip
      \hskip1.5em plus0em minus0.5em \fi \nobreak
      \vrule height0.75em width0.5em depth0.25em\fi}

\title{Below All Subsets for Some Permutational Counting Problems}

\author{Andreas Bj\"orklund \\ Department of Computer Science, Lund University}
\date{}
\begin{document}

\maketitle 
\begin{abstract}
We show that the two problems of computing the permanent of an $n\times n$ matrix of $\operatorname{poly}(n)$-bit integers and counting the number of Hamiltonian cycles in a directed $n$-vertex multigraph with $\operatorname{exp}(\operatorname{poly}(n))$ edges can be reduced to relatively few smaller instances of themselves. In effect we derive the first deterministic algorithms for these two problems that run in $o(2^n)$ time in the worst case. Classic $\operatorname{poly}(n)2^n$ time algorithms for the two problems have been known since the early 1960's.
Our algorithms run in $2^{n-\Omega(\sqrt{n/\log n})}$ time.
\end{abstract}

\renewcommand{\labelenumi}{\alph{enumi}.}

\section{Introduction}
We show that two well-known computationally hard counting problems defined over permutations, admit a strong form of self-reducibility. The problems are

\begin{itemize}
\item \textsc{Permanent}: Given an $n\!\times\! n$ matrix $M$ with $\operatorname{poly}(n)$--bit integer elements, compute $\operatorname{per}(M)=\sum_{\sigma \in S_n} \prod_{i} M_{i,\sigma(i)}$ where $S_n$ is the set of all permutations on $n$ elements.
\item \textsc{HamCycles}: Given an $n$-vertex directed multigraph, compute its number of Hamiltonian cycles, i.e. the number of non-crossing spanning cycles.
\end{itemize}

For both problems, we show that the solution to an instance of size parameter $n$ can be reduced to a weighted sum of the solutions to $\operatorname{poly}(n)2^{n-k}$ instances of size parameter $k<n$ of the same problem. Moreover, this reduction can be carried out in time polynomial in $n$ per generated instance.
We use this new relation to derive deterministic $2^{n-\Omega(\sqrt{n/\log n})}$ time algorithms for both \textsc{Permanent} and \textsc{HamCycles}. 
As a direct corollary we obtain an $Mn^22^{n-\Omega(\sqrt{n/log (Mn)})}+M^2n^4$ time algorithm for Asymmetric TSP in graphs with integer arc weights in $[0,\dots, M]$.

This is as far as the author knows the first deterministic algorithms that compute these quantities faster than explicitly inspecting at least a constant fraction of all subsets of an $n$-element set. In particular, no $o(2^n)$ time algorithms were previously known. 

Our techniques here are elementary and the presentation is more-or-less self-contained. The main components are inclusion--exclusion counting, polynomial interpolation, and the Chinese remainder theorem. The speed-up is obtained through tabulation.

The two problems have well-known $\operatorname{poly}(n)2^n$ time algorithms: Ryser's algorithm based on inclusion--exclusion for the permanent \cite{R63} from 1963, and a simple variation of Bellman, Held and Karp's dynamic programming algorithm for TSP \cite{B62,HK62} from 1962. Later a polynomial space inclusion--exclusion algorithm in the same spirit as Ryser's for counting Hamiltonian cycles with the same running time was found \cite{KGK77} in 1977 (and was rediscovered twice \cite{K82,B93}). 

The  question of existence of $O((2-\Omega(1))^n)$ time algorithms for the two problems are well-known open problems. In comparison the recent $O(1.657^n)$ time algorithm for Hamiltonian cycle \cite{B10} is randomized, only works for undirected graphs, and
 cannot even approximate the number of solutions. Very recently, Cygan et al.~\cite{C13} gave an algorithm for Hamiltonicity detection in bipartite directed graphs in $O(1.888^n)$ time. Still, not only have there been no deterministic algorithms running in $o(2^n)$ worst case time for the counting problems, it was not even known how to detect a Hamiltonian cycle in a directed $n$-vertex graph that fast, probabilistic algorithms included. Nor was it known how to compute the permanent of an $n\times n$ $0-1$ matrix deterministically in $o(2^n)$ time.
 
Moreover, Knuth asks in exercise 4.6.4.11. [M46] in \cite{K2} if it is possible to compute a real $n\times n$-matrix permanent with less than $2^n$ arithmetic operations. We note that reals of bounded precision can be modeled by large integers, so our algorithm here works also for them. However, a table look-up is not an arithmetic operation, so our algorithm is not exactly what Knuth solicited.

The one general previous improvement over $\operatorname{poly}(n)2^n$ time for any of the two exact counting problems we are aware of is the $2^{n-\Omega(n^{1/3}\log n)}$ \emph{expected} time algorithm for the $0-1$ matrix version of \textsc{Permanent} by Bax and Franklin \cite{BF02}.  Their technique can be extended to work with $O(1)$-bit integers, but probably not beyond that. In contrast, besides being faster and in deterministic time, our algorithm handles $\operatorname{poly}(n)$-bit integers, including negative ones.

The two known $\operatorname{poly}(n)2^n$ time algorithms for the problems based on the principle of inclusion--exclusion, Ryser's \cite{R63} and Kohn et al.'s \cite{KGK77} respectively, both use only polynomial space. It is indeed very natural to ask if employing the unexploited resource of using almost as much space as time wouldn't lead to faster algorithms. The problem though with the known approaches above is that there is no evident candidate for what to tabulate. They both sum over too large and typically different combinatorial objects. In the case of Ryser's permanent it is an $n$-element vector, and in Kohn et al.'s Hamiltonian cycles it is an induced graph on $n/2$ vertices on average.
 
The key insight here enabling a speed-up from tabulation is that the two problems admit a  mapping from the original instances down to a linear combination of not too many much smaller ones. So small in fact that they are bound to coincide, making tabulation worthwhile. 

\subsection{Overview of the Technique}
Consider the \textsc{Permanent} case, the \textsc{HamCycles} is similar.
The speed-up is obtained in a series of steps. 
First we let $k=c\sqrt{n/\log n}$ for a constant $c$ depending on the largest absolute element in the input matrix.
Next we employ the existence part of the Chinese remainder theorem to bring matrix elements down to $d\log n$ bits each for some $d$. That is, we compute the permanent modulo small primes $p$ of size polynomial in $n$. For each such prime $p$, 
we construct $\operatorname{poly}(n)2^{n-k}$ $k\times k$-matrices such that the permanent of the original one is equal to the sum of weighted permanents of all the matrices constructed. This reduction is in itself a two step procedure composed of a reduction to an inclusion--exclusion formula over polynomial matrices, accompanied by polynomial interpolation. We count the occurrences of each of the smaller matrices in a table. Next we compute the permanent once for each of the different smaller matrices appearing in the sum using the classic $\operatorname{poly}(k)2^k$ time algorithm.
We note that there are at most $n^{dk^2}<\!\!<2^n$ different such matrices of size $k\times k$.  The original instance permanent is then computed as a linear combination of all the tabulated matrices' permanent values. Finally, the results for all considered primes $p$ are assembled via the constructive part of the Chinese remainder theorem.

\subsection{Organization}
In Section~\ref{sec: sr} we give a self-contained description of the self-reduction, anticipating that this part of the results may be of independent interest. The main results, the $o(2^n)$ algorithms for the two counting problems, are described in Section~\ref{sec: alg}. 

\section{The Self-Reduction}
\label{sec: sr}
The two problems \textsc{Permanent} and \textsc{HamCycles} are closely related. At a first glance it appears that the first asks about a property of matrices and the second about graphs, but they can be expressed in the same language. For the purpose of this paper, we will redefine both the \textsc{Permanent} and the \textsc{HamCycles} problem in terms of arc-weighted complete directed graphs to stress their similarity. In the remainder of this paper, the graph $G_n=(V,A)$ will denote the complete directed graph on $n$ vertices $V$ labelled $1$ through $n$.

The set of all permutations on $n$ elements, denoted by $S_n$, can naturally be partitioned after the number of cycles the permutation describes: A permutation $\sigma\in S_n$ can be interpreted as a directed graph on $n$ vertices, labeled $1$ through $n$, with the arcs $i,\sigma(i)$ for all $i$. Every vertex has exactly one outgoing and one incoming arc, i.e. the graph is a set of disjoint cycles covering the vertices. We will with $S^1_n$ denote the subset of $S_n$ of permutations consisting of exactly one such cycle. Hence the permanent can be viewed upon as a sum over cycle covers of a graph, and the Hamiltonian cycles a sum over cycle covers consisting of just one cycle.

In the following it will make sense to be explicitly clear about what ring the computation is over. Thus we extend our problem definitions to
\begin{definition}[$R$-\textsc{Permanent}]
\label{def: per}
Given a complete directed graph $G_n=(V,A)$ and a function $f:A\rightarrow R$ mapping the arcs to some ring $R$, the \emph{permanent} of $(G,f)$ over $R$, denoted $\operatorname{per}(G,f)$, is 
$
\sum_{\sigma\in S_n} \prod_{i=1}^n f(i\sigma (i)).
$
\end{definition}
\begin{definition}[$R$-\textsc{HamCycles}]
\label{def: hc}
Given a complete directed graph $G_n=(V,A)$ and a function $f:A\rightarrow R$ mapping the arcs to some ring $R$, the \emph{hamcycles} of $(G,f)$ over $R$, denoted $\operatorname{hc}(G,f)$, is 
$
\sum_{\sigma\in S^1_n} \prod_{i=1}^n f(i\sigma (i)). 
$
\end{definition}

In the remainder of this section we will prove the following two lemmas:
\begin{lem}
\label{lem: per sr}
Given an instance $(G_n,f)$  to $F$-\textsc{Permanent} with $f$ mapping arcs to a field $F$ having at least $(n-k)n+1$ elements, and a positive integer $k<n$, one can compute $m=((n-k)n+1)2^{n-k}$ instances $I_i=(G_k,f_i)$ to $F$-\textsc{Permanent} and constants $a_i\in F$ for $i=1,\ldots, m$ ,  so that
\[
\operatorname{per}(G_n,f)=\sum_{i=1}^{m} a_i\operatorname{per}(G_k,f_i)
\]
Moreover, the constructed smaller instances and constants can be produced in polynomial in $n$ arithmetic operations $+$ and $*$ over $F$ per instance.
\end{lem}
\begin{lem}
\label{lem: hc sr}
Given an instance $(G_n,f)$ to $F$-\textsc{HamCycles} with $f$ mapping arcs to a field $F$ having at least $(n-k)k+1$ elements, and a positive integer $k<n$, one can compute $m=((n-k)k+1)2^{n-k}$ instances $I_i=(G_k,f_i)$ to $F$-\textsc{HamCycles} and constants $a_i\in F$ for $i=1,\ldots, m$ , so that
\[
\operatorname{hc}(G_n,f)=\sum_{i=1}^{m} a_i\operatorname{hc}(G_k,f_i)
\]
Moreover, the constructed smaller instances and constants can be produced in polynomial in $n$ arithmetic operations $+$ and $*$ over $F$ per instance.
\end{lem}

\subsection{Preliminaries}
In a complete directed graph $G_n$ a \emph{walk} of length $l$ is a sequence of not necessarily distinct vertices $v_0,v_1,\ldots, v_l$. If $v_0=v_l$ we say that the walk is a \emph{closed} walk.
For a field $F$ and an indeterminate $r$, we denote by $F[r]$ the polynomial ring over $F$ of polynomials in $r$ with coefficients from $F$. For a polynomial $p(r)\in F[r]$ we denote by $[r^n]p(r)$ the coefficient of the monomial $r^n$ in $p(r)$.

\subsection{Step 1. Inclusion--exclusion}
Consider an instance $(G_n,f)$ to either $F$-\textsc{Permanent} or $F$-\textsc{HamCycles} for some field $F$.
We fix a subset $K\subseteq V$ of the vertices of size $|K|=k$, called the \emph{kernel} of the reduction. Without loss of generality, we let $K$ be the vertices labeled by $1,2,\ldots,k$, and hence $V-K$ be the vertices labelled by $k+1,k+2,\ldots, n$. 

Our resulting instances will all be over the kernel $K$, i.e. embedded on the graph $G_k$. The central idea is to represent the parts of a cycle cover covering the vertices $V-K$, by arcs in $G_k$ between the entry and exit points of the cycles in $K$. This approach of representing parts of a cycle cover outside a small subgraph by encoding them on the arcs of the subgraph was previously used by the author both in \cite{B12} and \cite{B10}. The novelty here, is the observation that these reductions can be seen as a mapping to a low degree univariate polynomial, that in step 2 in the next section will be efficiently brought back to the original field.

In this first step, we construct one instance per subset of $V-K$, and use the principle of inclusion--exclusion to relate them to the original instance. The resulting instances will not be over the original field $F$ though. Instead the function $f$ giving weights to the arcs will assign polynomials in one rank indeterminate $r$ to them.

First we define the ranked walks in a vertex subset $X$.  The degree of the indeterminate $r$ counts the number of vertices visited along the walk.
For any vertices $u,v\in X\subseteq V$ we let $W_{X,k}(u,v)$ be the ranked walks between vertices $u$ and $v$ visiting $k$ vertices in $X$. We set
\begin{equation}
\label{eq: rw}
W_{X,k}(u,v)=\left\{\begin{array}{ll} \sum_{w\in X} W_{X,k-1}(u,w)f(w,v) r & : k>0 \\ 
	1 & : k=0 \wedge u=v \\
	0 & : k=0 \wedge u\neq v
	\end{array} \right.
\end{equation}

The ranked walks will be used to make sure all vertices outside the kernel $K$ are visited by the cycle covers in the \textsc{Permanent} case and the Hamiltonian cycles in the \textsc{HamCycles} case.
The principle of inclusion--exclusion makes sure crossing walks are cancelled. Since the \textsc{HamCycles} case is somewhat easier technically, we describe it first.

\subsubsection{Inclusion--exclusion for \textsc{HamCycles}}
We will construct instances to $F[r]$-\textsc{HamCycles} defined on $G_k=(K,A_K)$. We 
let $f_{X}:A_K \rightarrow F[r]$ for $X\subseteq V-K$ be defined for all $u,v\in K$ as follows

\begin{equation}
\label{eq: f_X}
f_{X}(uv)=f(uv)+\sum_{w,z\in X}f(uw) \left( \sum_{i=0}^{n-k-1}W_{X,i}(w,z) \right) f(zv)\cdot r.
\end{equation}

The point is that $f_{X}(uv)$ encodes all possible choices between either staying in $K$ by choosing the arc $uv$ directly or taking a detour through $V-K$ consisting of $1,2,\ldots,n-k$ vertices starting in $u$ and ending in $v$.

\begin{lem}
\label{lem: hc}
With $G_n,f,K,k,G_k,f_X$ as above it holds that
\[
\operatorname{hc}(G_n,f)=[r^{n-k}]\sum_{X\subseteq V-K }(-1)^{|V-K-X|} \operatorname{hc}(G_k,f_X) 
\]
\end{lem}
\begin{proof}
By the definition of $F$-\textsc{HamCycles} Def.~\ref{def: hc}, we have
\[
\operatorname{hc}(G_k,f_X) = \sum_{\sigma\in S_k^1} \prod_{i=1}^k f_X(i\sigma(i))
\]
Expanding $f_X$ via Eq.~\ref{eq: f_X}, we get
\[
\operatorname{hc}(G_k,f_X) = 
\sum_{\sigma\in S_k^1} \! \prod_{i=1}^k \! \left(f(i\sigma(i))\!+\!\sum_{l=1}^{n-k} r^l\!\!\!\!\!\!\sum_{v_1,\ldots,v_l \in X}\!\!\!\!\!\!\!\!\!f(iv_1) \left( \prod_{j=1}^{l-1}f(v_jv_{j+1}) \right) f(v_l\sigma(i))\!\right)
\]
From the formula above, we see that $[r^{n-k}]hc(G_k,f_X)$ is a sum over contributions
$
\prod_{i=1}^n f(v_iv_{i+1})
$
from closed walks $v_1,v_2,\ldots,v_{n+1}$ with $v_{n+1}=v_1$ where 
\begin{enumerate}
\item Exactly $n-k$ of $v_1,\ldots, v_n$ belong to $X$, and
\item Each vertex in $K$ occurs exactly once in $v_1,\ldots, v_n$.
\end{enumerate}

In the inclusion--exclusion summation over $X\subseteq V-K$,
\[
\operatorname{hc}(G_n,f)=\sum_{X\subseteq V-K} (-1)^{|V-K-X|}[r^{n-k}]\operatorname{hc}(G_k,f_X)
\]
each walk that crosses itself, i.e. has $v_i=v_j$ for some $i<j\leq n$, will be counted an even number of times. Moreover, exactly half of these times it will be added to the sum and the other half it will be subtracted, thereby canceling in the sum. To see why, let $Y=\{v_i|v_i\in V-K\}$ for a crossing walk. Clearly $Y\subset V-K$ since there are precisely $n-k$ vertices from $V-K$ on every contributing walk, and when one occurs at least twice there must be another one that is missing. Since among the subsets $Z$ fulfilling $Y\subseteq Z \subseteq V-K$ there are as many even sized subsets as odd ones the claim follows. Contributing walks that do not cross themselves however, i.e. are Hamiltonian cycles in $G$, will only be counted once, for $X=V-K$.
\qed
\end{proof}
\subsubsection{Inclusion--exclusion for \textsc{Permanent}}
In addition to the ranked walks in $V-K$ we also need to keep track of ranked cycles in $V-K$ for the \textsc{Permanent}. We want to sum over all cycle covers of the input graph $G$ and unlike the \textsc{HamCycles} case we may have vertices in $V-K$ disconnected from $K$ in a cycle cover. Remember that the vertices in $V$ are labelled $1,2,\ldots, n$ and associate the natural ordering $<$ of them. We need to define cycles in a cycle cover so that they receive a unique identifier to avoid double counting in our polynomial identity. To this end, we use that every cycle has a minimum vertex under the ordering to define the ranked closed walks \emph{anchored} at $s\in X$ as
\begin{equation}
\label{eq: CX}
C_X(s)=1+\sum_{i=1}^{n-k} W_{X_{\geq s},i}(s,s)
\end{equation}
where $X_{\geq s}=\{v|s\leq v\in X\}$, i.e. all vertices in $X$ equal to or larger than $s$.
The cycles anchored at $s$ represents all cycles of length $1,2,\ldots,n-k$ in $V-K$ where $s$ is the smallest vertex on the cycle. Note in particular that self-loops through $s$ are also included in the sum. The $1$ is in the definition of Eq.~\ref{eq: CX} to take into account the possibility that no cycle is anchored at $s$ in a contributing cycle cover.

\begin{lem}
\label{lem: per}
With $G_n,f,K,k,G_k,f_X$ as above it holds that
\[
\operatorname{per}(G_n,f)=[r^{n-k}]\sum_{X\subseteq V-K}  (-1)^{|V-K-X|} \operatorname{per}(G_k,f_X)\prod_{s\in X} C_X(s) 
\]
\end{lem}
\begin{proof}
By the definition of $F$-\textsc{Permanent} Def.~\ref{def: per}, we have
\[
\operatorname{per}(G_k,f_X) \prod_{s\in X} C_X(s)=\sum_{\sigma\in S_k} 
\prod_{j=1}^k f_X(j\sigma(j)) \prod_{i=k+1}^n C_X(i)
\]
Expanding $C_X$ via Eq.~\ref{eq: CX} and $f_X$ via Eq.~\ref{eq: f_X}, we get
\begin{multline*}
\operatorname{per}(G_k,f_X) \prod_{s\in X} C_X(s)= \\
\sum_{\sigma\in S_k}\prod_{i=1}^k \! \left(f(i\sigma(i))\!+\!\sum_{l=1}^{n-k} r^l\!\!\!\!\!\!\sum_{v_1,\ldots,v_l \in X}\!\!\!\!\!\!\!\!\!f(iv_1) \left( \prod_{j=1}^{l-1}f(v_jv_{j+1}) \right) f(v_l\sigma(i))\!\right) \\
\cdot \prod_{i=k+1}^n \left( 1+\sum_{l=1}^{n-k} r^l\!\!\!\!\!\!\sum_{\substack{v_1,\ldots,v_l\in X_{\geq i}\\i=v_1}} f(v_lv_1)\prod_{j=1}^{l-1} f(v_{j}v_{j+1}) \right) 
\end{multline*}

Expanding the formula above into a sum--product formula by identifying terms, we see that 
\[
[r^{n-k}]\operatorname{per}(G_k,f_X)\prod_{s\in X} C_X(s)
\] 
is a sum over contributions
$
\prod_{i=1}^l \prod_{uv\in O_i} f(uv)
$
for $1\leq l \leq n$ closed walks $O_i=v_{i,1},v_{i,2},\ldots,v_{i,m_l},v_{i,m_l+1}$ with $v_{i,1}=v_{i,m_l+1}$ and $\sum_{i=1}^l m_i = n$ where
\begin{enumerate}
\item Exactly $n-k$ of the $v_{i,j}$ for $1\leq i\leq n,1\leq j\leq m_i$ belong to $X$, and
\item Each vertex in $K$ occurs exactly once in the closed walks $O_i,1\leq i\leq l$.
\end{enumerate}

In the inclusion--exclusion summation over $X\subseteq V-K$,
\[
\operatorname{per}(G_n,f)=\!\!\!\!\!\!\sum_{X\subseteq V-K} (-1)^{|V-K-X|}[r^{n-k}]\!\operatorname{per}(G_k,f_X)\prod_{s\in X}C_X(s)
\]
each set of closed walks $\{O_i\}$ that crosses itself, i.e. has $v_{i1,j1}=v_{i2,j2}$ for some $i1\neq i2 \vee j1\neq j2$, will be counted an even number of times. Moreover, exactly half of these times it will be added to the sum and the other half it will be subtracted, thereby canceling in the sum. To see why, again let $Y=\{v_{i,j}|v_{i,j}\in V-K\}$ for a set of closed walks with a crossing. Clearly $Y\subset V-K$ since there are precisely $n-k$ vertices from $V-K$ on every contributing set of closed walks, and when one occurs at least twice there must be another one that is missing. Since among the subsets $Z$ fulfilling $Y\subseteq Z \subseteq V-K$ there are as many even sized subsets as odd ones the claim follows. Contributing sets of closed walks that do not cross themselves, i.e. are cycle covers in $G$, will only be counted once, for $X=V-K$.
\qed
\end{proof}
\subsection{Step 2. Polynomial Interpolation}
In the previous section we related the permanent and the Hamiltonian cycles of an arc weighted graph to smaller graphs with weights over a polynomial ring. We want to bring the small instances to map arcs to the original ring to complete the self-reduction. Unfortunately, we are only able to do this if the original ring is a field, and one that has at least polynomially many elements in the original instance size parameter. In particular, we need the following well-known result:

\begin{lem}[Lagrange interpolation]
For any set of pairs $\{(r_i,s_i)\}$ with distinct $r_i$'s and $r_i,s_i\in F$ for $i=1,\ldots,k+1$ where $F$ is a field on at least $k+1$ elements, there is a unique polynomial $p(r)$ in $F[r]$ of degree $k$ such that
$p(r_i)=s_i$ for all $i$. Moreover, the polynomial is given by 
\[ 
p(r)=\sum_{i=1}^{k+1} s_i\prod_{j\neq i} \frac{r-r_j}{r_i-r_j}
\]
\end{lem}

Specifically, consider an instance $(G,f)$ to $F$-\textsc{HamCycles}. Via Lemma~\ref{lem: hc} we see that $\operatorname{hc}(G)$ is related to a coefficient in a polynomial sum of many smaller instances $(G_k,f_X)$ to $F[r]$-\textsc{HamCycles}. We use here that if we know the result in enough points over $F$ we can reconstruct the polynomial via interpolation.

\begin{lem}
\label{lem: hc int}
For every polynomial term $\operatorname{hc}(G_k,f_X)$ in the outer sum in Lemma~\ref{lem: hc}, it is possible to compute $(n-k)k+1$ instances $(G_k,f_i)$ for $i=1,\ldots, (n-k)k+1$ to the $F$-\textsc{HamCycles} on $k$ vertices, and constants $a_i\in F$ for $i=1,\ldots, (n-k)k+1$ so that 
\[
|r^{n-k}]\operatorname{hc}(G_k,f_X)=\sum_{j=1}^{(n-k)k+1} a_j \operatorname{hc}(G_k,f_j)
\]
\end{lem}
\begin{proof}
Each entry in the codomain of $f_X$ has degree $n-k$ in $r$ by definition of the ranked walks and the definition of $f_X$ in Eq.~\ref{eq: f_X}. Since $\operatorname{hc}(G_k,f_X)$ is a sum over the product of $k$ arcs' $f_X$'s, the degree of $\operatorname{hc}(G_k,f_X)$ in $r$ is $(n-k)k$.  

Let $r_1,r_2,\ldots,r_m$ be $m$ distinct elements in $F$ and let $f_j$ be equal to $f_X$ evaluated in $r=r_j$. By Lagrange interpolation, it is possible to compute $\operatorname{hc}(G_k,f_X)$ and in particular the coefficient of $r_{n-k}$ from the evaluated polynomial points $\operatorname{hc}(G_k,f_j)$.
\qed
\end{proof}

The $F$-\textsc{Permanent} case is similar: consider an instance $(G_n,f)$. Lemma~\ref{lem: per} states that $\operatorname{per}(G_n,f)$ is related to a coefficient in a polynomial resulting from a sum of many smaller instances $(G_k,f_X)$ to $F[r]$-\textsc{Permanent}.
\begin{lem}
\label{lem: per int}
For every polynomial term $\operatorname{per}(G_k,f_X)\prod_{i=k+1}^{n} C_X(i)$ in the outer sum in Lemma~\ref{lem: per}, it is possible to compute $(n-k)n+1$ instances $(G_k,f_i)$ for $i=1,\ldots, (n-k)n+1$ to the $F$-\textsc{Permanent} on $k$ vertices, and constants $a_i\in F$ for $i=1,\ldots, (n-k)n+1$ so that 
\[
|r^{n-k}]\operatorname{per}(G_k,f_X)\sum_{i=k+1}^n C_X(i)=\sum_{j=1}^{(n-k)n+1} a_j \operatorname{per}(G_k,f_j)
\]
\end{lem}
\begin{proof}
Each entry in the codomain of $f_X$ has degree $n-k$ by definition of the ranked walks and the definition of $f_X$ in Eq.~\ref{eq: f_X}. Since $\operatorname{per}(G_k,f_X)$ is a sum over the product of $k$ arcs $f_X$'s, the degree of $\operatorname{per}(G_k,f_X)$ in $r$ is $(n-k)k$. The degree of $\prod_{i=k+1}^{n} C_X(i)$ is $(n-k)(n-k)$ since every $C_X(i)$ has degree $n-k$ by the definition Eq.~\ref{eq: CX}. Altogether, $\operatorname{per}(G_k,f_X)\prod_{i=k+1}^{n} C_X(i)$ has degree $(n-k)n$.

Let $r_1,r_2,\ldots,r_m$ be $m$ distinct elements in $F$ and let $f_j$ be equal to $f_X$ evaluated in $r=r_j$. Likewise, let $b_j$ be equal to $\prod_{i=k+1}^{n} C_X(i)$ evaluated in $r=r_j$.
By Lagrange interpolation, it is possible to compute the coefficent of $r^{n-k}$ in $\operatorname{per}(G_k,f_X)\prod_{i=k+1}^{n} C_X(i)$ from the evaluated polynomial points $b_j\operatorname{per}(G_k,f_j)$.
\qed
\end{proof}

The self-reduction for $F$-\textsc{Permanent} Lemma~\ref{lem: per sr} follows from the combination of Lemma~\ref{lem: per} and Lemma~\ref{lem: per int}, after observing that each $X\subseteq V-K$ and each $r\in 1,\ldots,(n-k)n+1$ corresponds to one small instance. Similarly, the self-reduction for $F$-\textsc{HamCycles} Lemma~\ref{lem: hc sr} follows from Lemma~\ref{lem: hc} and Lemma~\ref{lem: hc int} with  $X\subseteq V-K$ and $r\in 1,\ldots,(n-k)k+1$. 
It remains to validate the runtime in terms of the number of arithmetic operations used. To compute a small instance $(G_k,f_i)$ in Lemma~\ref{lem: per sr} (Lemma~\ref{lem: hc sr} respectively), corresponding to a particular $X\subseteq V-K$ and $r\in 1,\ldots,(n-k)n+1$, we see from the definitions Eqs.~\ref{eq: f_X} and~\ref{eq: CX} that the instance elements are computed as walks in $X$ for a fixed $r$. We can compute the elements through the recursive definition of the ranked walks Eq.~\ref{eq: rw} via dynamic programming in only polynomial in $n$ number of arithmetic operations. 

\section{The Algorithms}
\label{sec: alg}
In this section we prove our main theorems:
\begin{thm}
\label{thm: per}
Any single $n\times n$ matrix instance of \textsc{Permanent} with $\operatorname{poly}(n)$-bit integer elements can be solved deterministically in $2^{n-\Omega(\sqrt{n/\log{n}})}$ time.
\end{thm}
\begin{thm}
\label{thm: hc}
Any single $n$-vertex directed graph instance of \textsc{HamCycles} with $\operatorname{exp}(\operatorname{poly}(n))$ number of arcs can be solved deterministically in $2^{n-\Omega(\sqrt{n/\log{n}})}$ time.
\end{thm}
We immediately observe that the above theorem via a standard embedding of the $(\operatorname{min},+)$-semiring on the integers, and polynomial interpolation. That is, we introduce yet another indeterminate $z$, associate an arc of  weight $w$ with $z^w$, and finally solve for the smallest non-zero monomial in the resulting polynomial, see e.g. \cite{KGK77}. Since the evaluated polynomial is of degree at most $Mn^2$, we get
\begin{cor}
The shortest Asymmetric Traveling Salesman Problem route in an $n$-vertex graph with integer arc weights in $[0,\dots,M]$ can be computed in $Mn^22^{n-\Omega(\sqrt{n/\log (Mn)})}+M^2n^4$ time.
\end{cor}

On the top level, the idea of the algorithms is to bring the computations down to small finite fields. We next use the self-reductions from Section~\ref{sec: sr} to transform the input matrix/graph down to so small ones that several of them will be identical. By tabulating which ones of them have been constructed in this process and how often, it then suffices to compute the permanent of the small matrices/the Hamiltonian cycles of the small graphs only once. To make this precise we first need some elementary results from number theory.

\subsection{Preliminaries on Modular Arithmetic}
The well-known Chinese remainder theorem has two parts, an existence and a constructive one. The existence part states that an integer solution to a set of linear modular equations is uniquely defined in the range between zero and the least common multiple of the moduli. The constructive part describes how to recover the solution given the modular equations. We state them here in a slightly modified form as we will need them
\begin{lem}[CRT]
\label{lem: crt}
Given $m$ distinct primes $p_i$, and residues $0\leq a_i<p_i$, $1\leq i \leq m$, 
\begin{itemize}
\item{Existence:}
There is a unique integer $n$ in 
$
-\left\lfloor \frac{\prod_{i=1}^m p_i}{2} \right\rfloor \leq n < \left\lceil \frac{\prod_{i=1}^m p_i}{2} \right\rceil$
fulfilling $n\equiv a_i (\mbox{ mod }p_i),1\leq i\leq m$. 
\item{Construction:}
$n$ can be computed by evaluating $n_+=\sum_{i=1}^m a_ir_i$ where
$
r_i=\prod_{j\neq i} p_j \left(\left(\prod_{j\neq i} p_j\right)^{-1} (\mbox{ mod } p_i)\right) 
$
and then setting 

$n=n_+$ if $n_+<\frac{\prod_{i=1}^m p_i}{2} $, and $n=n_+-\prod_{i=1}^m p$ otherwise.
\end{itemize}
\end{lem}

We also use the following bound of the prime number theorem to answer how many and large primes we will need to break down a computation using the CRT:
\begin{lem}[Rosser~\cite{R41}]
\label{lem: pnt}
For every integer $n\geq 55$ the number of primes $\pi(n)$ less than or equal to $n$ obey $n/(ln(n)+2)<\pi(n)<n/(ln(n)-4)$.
\end{lem}

\subsection{The Algorithm}
We will first describe the algorithm for the \textsc{Permanent} case Thm.~\ref{thm: per} , and then point out the few changes needed for the \textsc{HamCycles} case Thm.~\ref{thm: hc}. 
We begin by describing the algorithm in pseudo-code in Algorithm \ref{alg: per}. Next we will  explain the steps in more detail.
\begin{algorithm}
\caption{\texttt{Permanent} $\operatorname{per}(G_n,f)$}
\label{alg: per}
\begin{algorithmic}[1]
\STATE Let $M$ be the largest absolute value in the image of $f$.
\STATE Let $P$ be the smallest set of primes $>n^2$ such that $\prod_{p\in P} p> 2Mn!$.
\STATE Let $k=\lfloor \sqrt{.99n/\log_2 {p_{\mbox{max}}}} \rfloor$ where $p_{\mbox{max}}=\max_{p\in P} p$.
\FOR {each prime $p\in P$}
\STATE Initialize a table $T$ from all $\mathbb{Z}_p^{k\times k}$ matrices to the positive integers with all zeros.
\STATE Evaluate $f_{(p)}=f(\mbox{ mod } p)$.
\STATE Compute $m=(n-k)n2^{n-k}$ instances $(G_k,f_j)$ and constants $a_j$ for $j=1,\ldots,m$ to $\mathbb{Z}_p$-\textsc{Permanent} such that $\operatorname{per}(G_n,f_{(p)}) = \sum_{j=1}^m a_i\operatorname{per}(G_k,f_j)$.
\FOR {$j=1,\ldots,m$}
\STATE Let $T(f_j)=T(f_j)+a_j$ i.e. increase the entry in $T$ for the matrix represented by $f_j$ with $a_i$.
\ENDFOR
\STATE Set $sum=0$.
\FOR {each $g$ with non-zero table entry $T(g)$}
\STATE Compute $\operatorname{per}(G_k,g)$ using Ryser's permanent algorithm.
\STATE Let $sum=sum+T(g)\operatorname{per}(G_k,g) (\mbox{ mod }p)$.
\ENDFOR
\STATE Store $per(G_n,f_{(p)})=sum$ 
\ENDFOR
\STATE Compute the permanent over $\mathbb{Z}$ using the stored $\operatorname{per}(G_n,f_{(p)})$ for all $p\in P$ using the constructive part of CRT.
\end{algorithmic}
\end{algorithm}

The existence part of CRT Lemma~\ref{lem: crt} makes it clear that to compute an integer function solely with the operations $+$ and $*$ over the integers, one can just as well compute it modulo several primes and assemble the result in the end. Both the \textsc{Permanent} and the \textsc{HamCycles} problems are defined as sum--products, so to compute their quantities modulo a prime $p$, we can replace the input integers with their residues modulo $p$. Steps 2,4,6 of the algorithm do precisely that, transform the input integer \textsc{Permanent} instance to instances to $\mathbb{Z}_p$-\textsc{Permanent} for primes $p$. Step 7 next generates $(n-k)n2^{n-k}$ instances to the $\mathbb{Z}_p$-\textsc{Permanent} problem using the constructive proof for Lemma~$\ref{lem: per sr}$. Steps 8-10 counts the occurrences of each of the different matrices in $\mathbb{Z}_p^
{k\times k}$ by keeping track of the total coefficients of each of the smaller matrices' permanents in Lemma~\ref{lem: per sr}. Steps 11-16 computes the solution to the $n\times n$-matrix permanent $\operatorname{per}(G_n,f_{(p)})$, and finally step 18 assembles the modular results using the constructive part of the CRT Lemma~\ref{lem: crt}. The correctness of the algorithm follows from Lemma~\ref{lem: crt} and the self-reduction Lemma~\ref{lem: per sr},  after noting that enough primes are chosen in step 2.

To bound the runtime, the only question is how many and how large primes are required, and indirect, how large tables will be used? The permanent is a sum of $n!$ products of $n$ elements from the input function $f$. In step $2$ of the algorithm we measure the absolute max over all elements used to conclude that $|\operatorname{per}(G_n,f)|\leq M^nn!<2^{n^c}$ for some positive constant $c$ when the input entries have $\operatorname{poly}(n)$ bits. From Lemma~\ref{lem: pnt} we see that there are at least $m=n^d/(dln(n)+2)-n^2/(2ln(n)-4)$ primes larger than $n^2$ but smaller than $n^d$ for $n\geq 55$. We want the product of the first $m$ primes larger than $n^2$, the set of primes $P$ in step 4 of the algorithm, to be larger than $2\cdot 2^{n^c}$, i.e. $n^{2m}>2^{n^c+1}$. It is straightforward to note that a constant $d$ depending on $c$ will suffice, in fact using $d=c+3$ is more than enough. Hence $p_{\mbox{max}}$ in step 3 is bounded by $n^d$ for $d$ constant and $k$ is $\Omega(\sqrt{n/\log n})$.
 
For each prime $p\in P$ in step 4, we use a table $T$ in step 5-14 with one entry per matrix in $\mathbb{Z}_p^{k\times k}$. 
An upper bound on the number of matrices in $\mathbb{Z}_p^{k\times k}$ using $p_{\mbox{max}}$ from step 3 and $k$ from step 4 of the algorithm is $(\log_2 p_{\mbox{max}})^{k^2}<2^{0.99n}$. 
The runtime of steps 5-10 is easily seen to be $O((n-k)n2^{n-k})$ from the bound on the table $T$'s size and Lemma~\ref{lem: per sr}.
Computing  the permanent of each of the matrices is a $O(k2^k)$ time task with Ryser's algorithm~\cite{R63}, so the total runtime of steps 11-16 is $o(2^{n-k})$. Altogether, the loop at steps 4-17 is run a polynomial number of times, and step 18 is polynomial time, so we get $\operatorname{poly}(n)2^{n-k}$ time in total which is $2^{n-\Omega(\sqrt{n/\log n})}$ time as claimed.

To adjust the algorithm and the proof to \textsc{HamCycles}, all we need to do is to replace Lemma~\ref{lem: per sr} for Lemma~\ref{lem: hc sr} in step 7 of the algorithm and the analysis, and exchange Ryser's algorithm for the permanent in step 13 for e.g. Bax's \cite{B93} Hamiltonian cycle counting algorithm.
 
 \section*{Acknowledgments}
 I thank Thore Husfeldt, Alexander Golovnev, Petteri Kaski, Mikko Koivisto, and Ryan Williams and several anonymous referees for comments on an earlier draft of the paper and stimulating discussions on the subject.


\begin{thebibliography}{10}
\bibitem{B93} E.~T.~Bax. Inclusion and Exclusion algorithm for the Hamiltonian Path problem. Inform.~Process.~Lett. 47(4), pp. 203--207, 1993.
\bibitem{BF02} E.~Bax and J.~Franklin. A Permanent Algorithm with $exp[½ ( n ^{1/3}/2 ln(n) )]$ Expected Speedup for 0-1 Matrices. Algorithmica,
Volume 32, Number 1, 157--162, 2002.
\bibitem{B62} R.~Bellman. Dynamic programming treatment of the travelling salesman problem, J. Assoc. Comput. Mach. 9, pp. 61--63, 1962.
\bibitem{B12} A.~Bj\"orklund. Counting Perfect Matchings as fast as Ryser. Proceedings of the 23rd SODA, pp. 914--921, 2012 
\bibitem{B10} A.~Bj\"orklund. Determinant Sums for Undirected Hamiltonicity. Proceedings of the 51st FOCS, pp. 173--182, 2010.
\bibitem{C13} M. Cygan, S. Kratsch, and J. Nederlof. Fast Hamiltonicity checking via bases of perfect matchings. arXiv:1211.1506, 2012.
\bibitem{HK62} M.~Held and R.~M.~Karp. A dynamic programming approach to sequencing problems, J. Soc. Indust. Appl. Math. 10, pp. 196--210, 1962.
\bibitem{K82} R.~M.~Karp. Dynamic programming meets the principle of inclusion and exclusion, Oper.~Res.~Lett. 1 no. 2, pp. 49--51, 1982.
\bibitem{K2} D.~E.~Knuth. The Art of Computer Programming. Vol. 2: Seminumerical Algorithms, Third Edition,  (Reading, Massachusetts: Addison-Wesley, 1997), xiv+762pp. ISBN 0-201-89684-2, 1997.
\bibitem{KGK77} S.~Kohn, A.~Gottlieb, and M.~Kohn. A generating function approach to the traveling salesman problem, Proceedings of the 1977 Annual Conference (ACM'77), Association for Computing Machinery, pp. 294--300, 1977.
\bibitem{R41} B.~Rosser. Explicit Bounds for Some Functions of Prime Numbers. American Journal of Mathematics 63 (1): 211--232, 1941.
\bibitem{R63} H.~J.~Ryser.  Combinatorial Mathematics, The Carus mathematical monographs, The Mathematical Association of America, 1963.

\end{thebibliography}
\end{document}